\documentclass[11pt,a4paper]{article}
\usepackage{lineno}
\usepackage[margin=2cm]{geometry}
\usepackage[utf8]{inputenc}
\usepackage{amsmath}
\usepackage{amssymb}
\usepackage{amsthm}
\usepackage{dsfont}
\usepackage{color}
\usepackage{amsmath}
\usepackage{amsfonts}
\usepackage{amssymb}
\usepackage{bbm}
\usepackage{relsize}
\usepackage{tikz}
\usepackage[procnumbered, linesnumbered,algoruled]{algorithm2e}
\usepackage[compact]{titlesec}
\usepackage{amsmath}
\usepackage[hidelinks,colorlinks,allcolors=blue]{hyperref}
\usepackage[T1]{fontenc}
\usepackage[utf8]{inputenc}
\usepackage{epic, eepic}
\usepackage{epsf}

    \usepackage[utf8]{inputenc}

\newcommand{\ceil}[1]{{\left\lceil#1  \right\rceil}}
\newcommand{\comment}[1]{}

\newcommand{\ci}{{\mathcal{I}}}
\newcommand{\cA}{{\mathcal{A}}}

\newcommand{\cI}{{\mathcal{I}}}

\newcommand{\OPT}{\textnormal{OPT}}

\newcommand{\eps}{{\varepsilon}}

\DeclareMathOperator*{\argmax}{arg\,max}

\begin{document}
\newtheorem{thm}{Theorem}[section]
\newtheorem{prop}[thm]{Proposition}
\newtheorem{assm}[thm]{Assumption}
\newtheorem{conj}[thm]{Conjecture}
\newtheorem{lem}[thm]{Lemma}
\newtheorem{observation}[thm]{Observation}
\newtheorem{cor}[thm]{Corollary}
 \newtheorem{lemma}[thm]{Lemma}
  \newtheorem{definition}[thm]{Definition}
 \newtheorem{theorem}[thm]{Theorem}
 \newtheorem{proposition}[thm]{Proposition}
 \newtheorem{claim}[thm]{Claim}
\newtheorem{defn}[thm]{Definition}
\newcommand{\ariel}[1]{{\color{red} (Ariel :#1)}}
\def \II   {{\mathcal I}}
\newcommand{\one}{\mathbbm{1}}

%\begin{titlepage}
\title{
%	A Faster %$(1-e^{-1})$-
%	Tight Approximation for Submodular Maximization Subject to a Knapsack Constraint
%Tight Lower and Upper Bounds for
%Knapsack with Bipartite Conflicts
Tight Bounds for Budgeted Maximum Weight Independent Set in Bipartite and Perfect Graphs

}

\author{Ilan Doron-Arad\thanks{Computer Science Department, 
		Technion, Haifa, Israel. \texttt{idoron-arad@cs.technion.ac.il}}
	\and 
	Hadas Shachnai\thanks{Computer Science Department, 
		Technion, Haifa, Israel. \texttt{hadas@cs.technion.ac.il}}
}

%\author{ {Ilan Doron-Arad}\thanks{Computer Science Department, 
%		Technion, Haifa, Israel. \mbox{E-mail: {\tt idoron-arad@cs.technion.ac.il}}}
%	\and 
%	{Hadas Shachnai}\thanks{ Computer Science Department, 
%		Technion, Haifa, Israel. \mbox{ E-mail: {\tt hadas@cs.technion.ac.il}}}}  
\date{} 
\maketitle

\begin{abstract}
	We consider the classic {\em budgeted maximum weight independent set (BMWIS)} problem. The input is a graph $G = (V,E)$,  a weight function $w:V \rightarrow \mathbb{R}_{\geq 0}$, a cost function $c:V \rightarrow \mathbb{R}_{\geq 0}$, and a budget $B \in \mathbb{R}_{\geq 0}$. The goal is to find an independent set $S \subseteq V$ in $G$ such that $\sum_{v \in S} c(v) \leq B$, which maximizes the total weight $\sum_{v \in S} w(v)$. Since the problem on general graphs cannot be approximated within ratio $|V|^{1-\eps}$ for any $\eps>0$, BMWIS has attracted significant attention on graph families for which a maximum weight independent set can be computed in polynomial time. Two notable such graph families are bipartite and perfect graphs. BMWIS is known to be NP-hard on both of these graph families; however, the best possible  approximation guarantees for these graphs are wide open.
	%This includes the fundamental bipartite graphs and perfect graphs. Nevertheless, 
	%besides being strongly NP-hard, to-date it is unknown whether BMWIS on bipartite graphs admits a PTAS; on the other hand, it is also obscured if BMWIS on perfect graph admits a constant factor approximation. 
	
In this paper, we give 
%Our main result is
a tight $2$-approximation for BMWIS on perfect graphs and bipartite graphs.
In particular, %answer both of these questions by giving 
we give 
We a $(2-\eps)$ lower bound for BMWIS on bipartite graphs, already for the special case where the budget is replaced by a cardinality constraint, based on the {\em Small Set Expansion Hypothesis (SSEH)}. For the upper bound, we design a $2$-approximation for BMWIS on perfect graphs using
%using a simple reduction to
 a {\em Lagrangian relaxation} based technique. 
% We further 
Finally, we obtain 
%We complement the above with 
a tight lower bound for  the {\em capacitated maximum weight independent set (CMWIS)} problem, the special case of BMWIS where $w(v) = c(v)~\forall v \in V$. We show that CMWIS on bipartite and perfect graphs is unlikely to admit an
{\em efficient polynomial-time approximation scheme (EPTAS)}. 
%thus, essentially 
%This essentially implies that 
Thus, the existing PTAS for CMWIS is essentially the best we can expect.
% possible.  %This matches the natural $2$-approximation achieved by returning the optimal solution over one side of the graph.  
    %We show that for small enough $\eps>0$ the {\em knapsack problem with bipartite conflict graph (KPB)} is hard to approximate to a ratio of $\frac{15}{16}+\eps$ under the assumption stated in~\cite{feige2002relations} (R3SAT): Even when $\Delta$ is an arbitrarily large constant independent of $n$, there is no polynomial time algorithm that refutes most 3CNF formulas with $n$ variables and $m = \Delta n$ clauses, and never wrongly refutes a satisfiable formula.  In particular, this strengthen the results of Pferschy and Schauer~\cite{pferschy2009knapsack} and shows that there is no PTAS for the knapsack problem with perfect conflict graphs.
\end{abstract}

\section{Introduction}

We consider the {\em budgeted maximum weight independent set (BMWIS)} problem, defined as follows.\footnote{In \cite{pferschy2009knapsack,pferschy2017approximation} and several other papers the BMWIS problem is named {\em knapsack with a conflict graph (KCG)}.} The input is a graph $G = (V,E)$,  a weight function $w:V \rightarrow \mathbb{R}_{\geq 0}$, a cost function $c:V \rightarrow \mathbb{R}_{\geq 0}$, and a budget $B \in \mathbb{R}_{\geq 0}$. A {\em solution} is an independent set $S \subseteq V$ in $G$ such that $c(S) = \sum_{v \in S} c(v) \leq B$.\footnote{For any function $f:A \rightarrow \mathbb{R}_{\geq 0}$ and a subset of elements $C \subseteq A$, we define $f(C) = \sum_{e \in C} f(e)$.} The {\em weight} of a solution $S$ is $w(S) = \sum_{v \in S} w(v)$, and the goal is to find a solution of maximum weight. %We also consider  the {\em capacitated maximum weight independent set (CMWIS)}, the special case of BMWIS where $w(v) = c(v)$ for all $v \in V$. 
For short, we use MS for BMWIS, MSP for MS where the given graph $G$ is perfect, and MSB for MSP if $G$ is also bipartite (a special case of perfect graph).    

MS is a natural generalizations of the classic {\em maximum weight independent set (MWIS)} problem (where $B = \infty$) and the $0/1$-{\em knapsack} problem (for $E = \emptyset$), which are two cornerstone problems in theoretical computer science.
% that have many real-world applications. In particular, 
MS finds applications in {\em job scheduling} \cite{kleinberg2006algorithm} as well as in selecting a {\em non-interfering set of transmitters}~ \cite{andrews2007overcoming,suh2008interference,viswanathan2003downlink} (see, e.g.,~\cite{bandyapadhyay2014variant} for other applications). 

In this paper, we study the approximability of MSP  and MSB.
Let $\OPT(I)$ be the value of an optimal solution for an instance $I$ of a maximization problem~$\Pi$. 
For $\alpha \geq 1$, we say that $\cA$ is an $\alpha$-approximation algorithm
for $\Pi$ if, for any instance $I$ of $\Pi$, $\cA$ outputs a solution of value at least  $\alpha \cdot  \OPT(I)$. Being a generalization of MWIS, there is no $|V|^{1-\eps}$-approximation for MS unless NP=ZPP \cite{hastad1996clique}. As a result, there has been an interesting line of research on MS for graph families on which MWIS admits polynomial-time algorithms \cite{pferschy2009knapsack,bandyapadhyay2014variant,kalra2017maximum,pferschy2017approximation}. While MSP and MSB are known to be strongly NP-hard~\cite{pferschy2009knapsack}, there is currently no definitive answer regarding the best approximation guarantees for these intriguing problems. In particular, the only lower bound for MSB is strong NP-hardness~\cite{pferschy2009knapsack} and no non-trivial upper bound for MSP is known. 

%However, besides that the MSP and MSB are strongly NP-hard \cite{pferschy2009knapsack}, to-date there is no clear answer on the approximation guarantee of these two interesting problems. 

%besides being strongly NP-hard, to-date it is unknown whether MS on bipartite graphs admits a PTAS; on the other hand, it is also obscured if MS on perfect graph admits a constant factor approximation. 

Our main result is a tight $2$-approximation for BMWIS on perfect graphs and bipartite graphs. We first improve
the lower bound for MSB from strong NP-hardness \cite{pferschy2009knapsack} to $2-\eps$. In our reduction, we are able to utilize the strong hardness result of the {\em maximal balanced biclique} problem \cite{manurangsi2017inapproximability}. %The main idea is to 
We first 
convert the maximal balanced biclique instance into an {\em unbalanced} biclique instance.
We then construct an MSB instance for which any solution that takes vertices only from one side of the graph achieves at most (roughly) a $2$-approximation. Our lower bound holds already for the special case of uniform costs, thus answering negatively an open question of~\cite{kalra2017maximum}.

\begin{theorem}
	\label{thm:KP}
	Assuming the {\em Small Set Expansion Hypothesis (\textnormal{SSEH})}, for any $0<\eps<1$ there is no $(2-\eps)$-approximation for \textnormal{MSB}. 
\end{theorem}

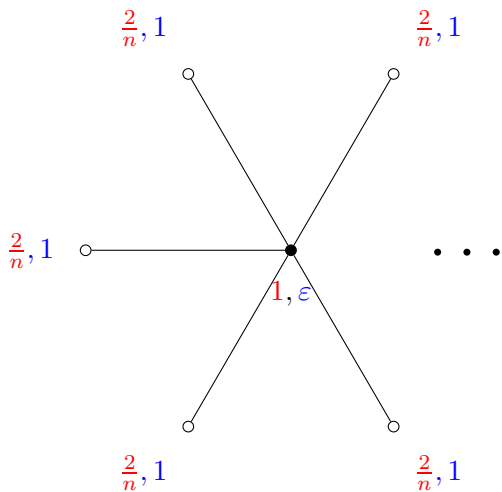
\begin{figure}

	\hspace{4cm}{
		\begin{tikzpicture}[scale=1.8]
			% Define the center node
			\node[circle, draw=black, fill=black, inner sep=0pt, minimum size=4pt, label={[label distance=0.2cm]270:{$\textcolor{red}{1},\textcolor{blue}{\varepsilon}$}}] (center) at (0,0) {};
			
			% Define the other nodes in the graph
			\foreach \i in {1,2,...,5}
			\node[circle, draw=black, fill=white, inner sep=1pt, minimum size=4pt, label={[label distance=0.2cm, blue]60*\i:{$\textcolor{red}{\frac{2}{n}},\textcolor{blue}{1}$}}] (node\i) at (60*\i:1.5cm) {};
			
			% Draw the edges between the center node and the other nodes
			\foreach \i in {1,2,...,5}
			\draw (center) -- (node\i);
			
			% Add a label to represent multiple nodes
			\node[circle, xshift = -40, draw=none, fill=none, inner sep=0pt, minimum size=4pt, label={[label distance=-0.4cm]300:{\Huge $\cdots$}}] (dots) at (00:1.8cm) {};
		\end{tikzpicture}
	}
	
	\caption{	\label{fig:1} An example for the difference between MWIS and BMWIS. In the figure there is a graph with $n+1$ vertices; each vertex has a weight (left number, in red) and a cost (right number, in blue). The optimal solution of the MWIS problem w.r.t. the weights is to take all vertices except the center vertex, which gives a total weight of $2$. For the MS problem with a budget of $B = 1$, the optimal solution is the center vertex. Any other vertex gives a solution of weight $\frac{2}{n}$ for the BMWIS problem.}
\end{figure}

To derive an approximation algorithm for MSP, one may be tempted to reduce the problem to an MWIS instance on a perfect graph, which admits a polynomial-time exact algorithm \cite{grotschel2012geometric}. More specifically, we can find an optimal solution $S$ for the induced MWIS instance obtained by removing the budget constraint. Then, to obtain a feasible solution for the MSP instance (i.e., satisfy the budget constraint) we keep a subset of vertices $T \subseteq S$ of maximum weight $w(T)$, such that the budget constraint is satisfied, i.e., $c(T) \leq B$. However, this idea may lead to an arbitrarily poor approximation guarantee for MSP. We give a simple example in Figure~\ref{fig:1}. 

Instead of the above approach, we reduce the MSP problem to a {\em subset selection} problem with a budget constraint~\cite{kulik2021lagrangian}, relying on the {\em Lagrangian relaxation} of the MSP problem. This, combined with Theorem~\ref{thm:KP}, gives the tight upper bound for MSP and MBS (recall that MSB is a special case of MSP). 

%CONT

%Observe that MBS is a special case of MSP. Our second result is a $2$-approximation for MSP based on a reduction to the technique of \cite{kulik2021lagrangian}, which resolves the complexity status of both MSP and MSB. 

 \begin{theorem}
	\label{thm:MSP}
	There is a $2$-approximation for \textnormal{MSP}.
\end{theorem}

 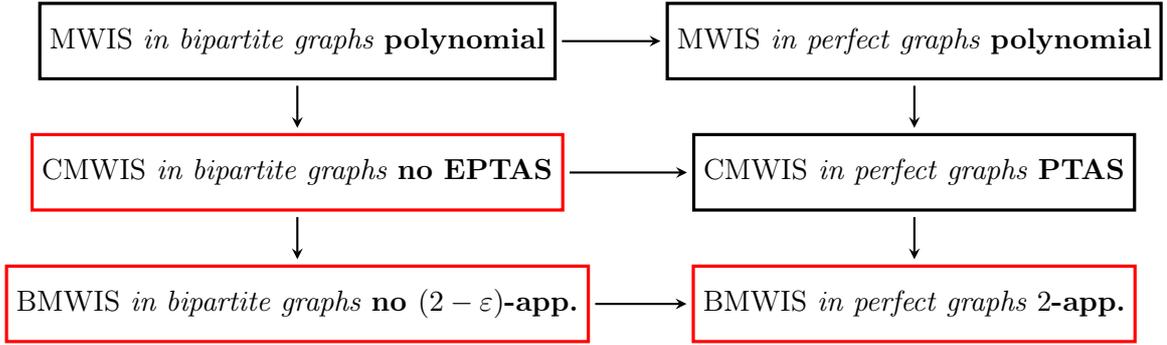
\begin{figure}\hspace*{0.8cm}
	\centering
	\label{fig:diagram}
	\begin{tikzpicture}[
		terminal/.style={
			% The shape:
			rectangle,
			% The size:
			minimum width=2cm,
			minimum height=1cm,
			% The border:
			very thick,
			draw=black,
			% Font
			font=\itshape,
		},
		]
		\matrix[row sep=0.7cm,column sep=1cm] {%
			% First row:
			; \node [terminal](b1) {\textnormal{MWIS} in bipartite graphs {\bf polynomial}}; &  \node [terminal](p1) {\textnormal{MWIS} in perfect graphs {\bf polynomial}};\\
			
			%	&; \node [terminal](KP) {$0/1$-knapsack {\bf FPTAS}};  &\\
			
			%			\node [terminal](C) {Subset Sum {\bf Weakly NP-Hard}}; &\\
			
			; \node [terminal,draw=red](b2) {\textnormal{CMWIS} in bipartite graphs {\bf no EPTAS}}; &  \node [terminal](p2) {\textnormal{CMWIS} in perfect graphs {\bf PTAS}};\\
			%	\node [terminal](C) {Subset Sum {\bf Weakly NP-Hard}}; &\\
			
			%	; \node [terminal,draw=red](KP) {\textnormal{MWIS} in bipartite graphs {\bf polynomial}}; &  \node [terminal](KP) \textnormal{BMWIS} in perfect graphs {\bf $2$-approximation}};\\
		; \node [terminal,draw=red](b3) {\textnormal{BMWIS} in bipartite graphs {\bf no $\left(2-\eps\right)$-app.}}; &  \node [terminal,draw=red](p3) {\textnormal{BMWIS} in perfect graphs {\bf $2$-app.}};\\
		%	\node [terminal,draw=red](C) {\textnormal{CMWIS} in bipartite graphs  {\bf no EPTAS} (this paper)}; &\\
		
		%	\node [terminal](C) {\textnormal{CMWIS} in perfect graphs  {\bf PTAS}}; &\\
		
		%	\node [terminal](MCK) {Subset Sum {\bf Weakly NP-Hard}}; & & \node [terminal](kKP) {Cardinality constrained knapsack {\bf FPTAS}}; \\%;   & \\
		%Second row
		%	&\node [terminal](PKP) {Knapsack with partition matroid {\bf FPTAS}}; &\\
		%\node [draw=none](BLM) {}; 
		\\
		\\
		%	& \node [terminal](BMI) {\textnormal{BMWIS} in perfect graphs {\bf $2$-approximation}}; \\
		% Third row:
		%\node [terminal](p2) {Cardinality Constrained Knapsack};  &\\
	};
	%	\draw   (p1) edge [->,>=stealth,shorten <=2pt, shorten >=2pt,thick] (p3)
	%	(p2) edge [->,>=stealth,shorten <=2pt, shorten >=2pt, thick] (p3);
	%	\draw   (p3) edge [->,>=stealth,shorten <=2pt, shorten >=2pt,thick] (p4);       
	%		\draw (KP) edge [->,>=stealth,shorten <=2pt, shorten >=2pt, thick] (kKP);
	\draw (b1) edge [->,>=stealth,shorten <=2pt, shorten >=2pt, thick] (b2);
	\draw (b1) edge [->,>=stealth,shorten <=2pt, shorten >=2pt, thick] (p1);
	\draw (b2) edge [->,>=stealth,shorten <=2pt, shorten >=2pt, thick] (p2);
	\draw (b3) edge [->,>=stealth,shorten <=2pt, shorten >=2pt, thick] (p3);
	\draw (p1) edge [->,>=stealth,shorten <=2pt, shorten >=2pt, thick] (p2);
	\draw (b2) edge [->,>=stealth,shorten <=2pt, shorten >=2pt, thick] (b3);
	\draw (p2) edge [->,>=stealth,shorten <=2pt, shorten >=2pt, thick] (p3);
\end{tikzpicture}
\vspace{-1.5cm}
\caption{\label{fig:sum} An overview on MWIS variants including our new tight bounds (in red), where {\em app.} is an abbreviation for approximation. An arrow from $A$ to $B$ indicates that $A$ is a special case of $B$.}
\end{figure}

We complement the above with a tight lower bound for the {\em capacitated maximum weight independent set (CMWIS)} on bipartite and perfect graphs; this is the special case of BMWIS where the weight of each vertex is equal to the cost $w(v) = c(v)~\forall v \in V$. We show that CMWIS on bipartite graphs is unlikely to admit an EPTAS; thus, the existing PTAS for CMWIS on general perfect graphs \cite{doron2023approximating} essentially cannot be improved. 

	\begin{theorem}
	\label{lem:EPTAS}
	Unless \textnormal{FPT=W[1]}, there is no \textnormal{EPTAS} for \textnormal{CMWIS} in bipartite graphs.\footnote{For more details on the relevant parametrized complexity assumptions, see, e.g., \cite{cygan2015parameterized}.}
\end{theorem}
See Figure~\ref{fig:sum} for a summary of our results. 
\subsection{Related Work}

In this section we focus on known approximation algorithms for MS and other variants of MWIS; for exact algorithms and heuristics see, e.g., \cite{coniglio2021new,basnet2018heuristics,bettinelli2017branch}. Pferschy and Schauer \cite{pferschy2009knapsack} were the first to study approximation algorithms for MS. They design pseudopolynomial algorithms and {\em fully polynomial time approximation schemes (FPTAS)} for MS on 
%important graph families such as 
graphs of bounded treewidth and chordal graphs. They also prove that MSB is strongly NP-hard, already for the special case of bipartite graphs with maximal degree $3$.

Kalra et al.~\cite{kalra2017maximum} consider the special case of MSB with uniform costs. They show that the problem is NP-hard and derive a $2$-approximation for this variant based on a greedy approach. Bandyapadhyay \cite{bandyapadhyay2014variant} also studied the MS problem. The paper includes a $d$-approximation for the special case of uniform weights MS on $(d+1)$-claw free graphs, and a {\em polynomial time approximation scheme (PTAS)} for MS on planar graphs. We note that the results of \cite{kulik2021lagrangian} can improve the approximation algorithm for MS on $d+1$-claw free graphs (see Section~\ref{sec:discussion}). For MS on {\em line graphs} (or, the {\em budgeted matching} problem) there is a PTAS~\cite{berger2011budgeted} and also an 
%{\em efficient PTAS (EPTAS)}
EPTAS~\cite{doron2023eptasMatching}. A similar variant of BMWIS considers the budgeted maximization of an independent set in a {\em matroid} \cite{chekuri2011multi,grandoni2010approximation,doron2023eptas,doron2023fptas}. 

The MWIS problem (i.e., MS without a budget constraint) is a well studied problem; we list below some notable works on this problem. Gr{\"o}tschel et al. \cite{grotschel2012geometric} show that MWIS on perfect graphs admits a polynomial-time algorithm. 
Considerable attention was given also to MWIS on planar graphs.
%Another graph family that attracts a significant attention is planar graphs. 
For this family, Baker \cite{baker1994approximation} presented a PTAS. For the family of $d$-claw free graphs, a $\left(  \frac{d}{2}+\eps \right)$-approximation  was proposed presented by Neuwohner~\cite{neuwohner2021improved}. Finally, for MWIS on graphs with maximal degree $\Delta$, Halld{\'o}rsson presented an elegant algorithm that achieves an approximation ratio of $\frac{\log \Delta}{\log \log \Delta}$~ \cite{halldorsson1998approximations}. 

We note that CMWIS is a generalization of MWIS; thus, on general graphs it
%Being a generalization of MWIS, CMWIS also 
cannot be approximated within a ratio better than $|V|^{1-\eps}$, unless NP=ZPP~\cite{hastad1996clique}. Furthermore, CMWIS remains (weakly) NP-hard even on 
degenerate graph which has no edges, 
%even if the conflict graph is degenerated and there are no edges CMWIS remains (weakly) NP-Hard 
as it generalizes the {\em subset sum} problem. The work of~\cite{doron2023approximating} presents a PTAS for CMWIS on perfect graphs; we are not aware of other results for this problem.

\noindent {\bf Organization:} In Section~\ref{sec:preliminaries} we give some notation and preliminary results. In Sections~\ref{sec:2} and~\ref{sec:3} we give the lower and upper bounds for MSB and MSP, respectively. In Section~\ref{sec:cap} we derive the lower bound for CMWIS. We conclude in Section~\ref{sec:discussion} with a discussion and open problems. 

\section{Preliminaries}
\label{sec:preliminaries}

\subsection{Independent Sets in Perfect Graphs}

Given a graph $G = (V,E)$, an {\em independent set} in $G$ is a subset of vertices $S \subseteq V$ such that for all $u,v \in S$ it holds that $(u,v) \notin E$. Let $\textsf{IS}(G)$ be the collection of all independent sets in a graph $G$. Given a weight function $w:V \rightarrow \mathbb{R}_{\geq 0}$, a {\em maximum weight independent set} is an independent set $S \in \textnormal{\textsf{IS}}(G)$ for which $w(S)$ is maximized. %In addition, {\em a coloring} of $G$ is a partition $(V_1, \ldots, V_t)$ of $V$ such that $\forall i \in [t]: V_i \in \textsf{IS}(G)$; we call each subset of vertices $V_i$ {\em color class} $i$. %We can thus obtain a valid coloring of $G$ with $t$ colors. 
A graph $G$ is called {\em perfect} if for every induced subgraph $G'$ of $G$ the cardinality of the maximum clique of $G'$ is equal to the minimum number of 
colors required to properly color $G'$ (i.e., the chromatic number $\chi(G')$). 
Note that $G'$ is also a perfect graph. %For a perfect graph $G = (V,E)$, the following can be computed in polynomial time by \cite{grotschel2012geometric}. % we assume that the following can be computed in polynomial time. 
The following well known result is due to \cite{grotschel2012geometric}. 	
\begin{lemma}  
\label{lem:grot}
Given a perfect graph $G = (V,E)$, a %minimum coloring of $G$ and 
maximum weight independent set of $G$ can be computed in polynomial time. 
\end{lemma}

\subsection{Budgeted Independent Set}
	%In the {\em knapsack with perfect conflicts (MSP)} problem, we are given 
	We use a tuple $\ci = (V,E,w,c,B)$ to denote an MSP instance, where $G = (V,E)$ is a perfect graph, $w:V \rightarrow \mathbb{R}_{\geq 0}$ is a weight function, $c:V \rightarrow \mathbb{R}_{\geq 0}$ is a cost function, and $B \in \mathbb{R}_{\geq 0}$ is a budget. 
	%A {\em solution} for $\ci$ is $S \subseteq V$ such that $c(S) \leq B$ and $S \in \textsf{IS}(G)$. The goal is to find a solution $S$ of $\ci$ such that $w(S)$ is maximized. The {\em budgeted independent set in bipartite graph (MSB)} problem  is the special case of MSP in which the graph $G$ is bipartite. 
	For convenience, we denote an instance of MSB by $\ci = (L,R,E,w,c,B)$, where $G = (L \cup R,E)$ is the given bipartite graph. The special case of MSB where $c$ is {\em uniform} is called below U-MSB. %called below MSB with {\em cardinality constraint (U-MSB)}. 
	Given a U-MSB instance, we assume w.l.o.g. that $\forall v \in L \cup R: c(v) = 1$. 

\section{The Lower Bound for Uniform MSB}
\label{sec:2}

	In this section we give the proof of the lower bound presented in Theorem~\ref{thm:KP}. We reduce the problem of identifying large {\em balanced bicliques} in a bipartite graph to MSB. As a starting point, we use the following strong hardness result due to \cite{manurangsi2017inapproximability}.\footnote{We refer the reader to \cite{manurangsi2017inapproximability} for further details on the {\em small set expansion hypothesis(SSEH)}.}  For $t,s \in \mathbb{N}$, let $K_{t,s}$ be a complete bipartite graph in which one side contains exactly $t$ vertices, and the other side contains exactly $s$ vertices. We say that $K_{t,s}$ is a {\em biclique} with parameters $t$ and $s$. If $t = s$ then $K_{t,s}$ is {\em balanced} and is {\em unbalanced} otherwise.

	\comment{
	Despite the centrality of MSB, prior to this work it was only known that the problem is NP-hard in the strong sense~\cite{pferschy2009knapsack}, already in the special case of unit costs, i.e., U-MSB. For this special case, the best known approximation ratio is $2$, that is easily achieved by solving optimally {\em knapsack with cardinality constraint} on each side of the bipartite graph and returning the better solution~\cite{kalra2017maximum}. The question whether U-MSB admits a $(2 - \eps)$-approximation remained open.
	
	In this work we answer this question negatively, by giving a $(2-\eps)$ lower bound for U-MSB based on the {\em Small Set Expansion Hypothesis (SSEH)}.
	This closes the approximation gap for U-MSB. Furthermore, we show (in  Section~\ref{sec:MSP_upper_bound}) a matching upper bound of $2$ for MSP. 
}

\begin{lemma} \cite{manurangsi2017inapproximability} 
	\label{lem:hardness1}
	Assuming SSEH, for any $\delta > 0$ and $n \in \mathbb{N}$, given a bipartite graph $G = (L,R,E)$ with $|L| = |R| = n$, it is NP-hard to distinguish between the following two cases: \begin{itemize}
		\item 	(Completeness, "yes" case) $G$ contains $K_{\left(\frac{1}{2}-\delta\right) \cdot n, \left(\frac{1}{2}-\delta\right) \cdot n}$ as a subgraph. 
		
		\item (Soundness, "no" case) $G$ does not contain $K_{ \delta \cdot n,  \delta \cdot n}$ as a subgraph. 
		
		%Here EH(T) , {e ∈ EH | e ⊆ T} denotes the set of hyperedges inside of the set T ⊆ VH.
	\end{itemize}

\end{lemma}

%We create vectors of vertices given a graph, and using an OR operator we c

We give the next lemma as an intermediate step in our hardness result. Intuitively, it states that the problem of finding an {\em unbalanced} biclique with suitable parameters is as hard as the problem of finding a balanced biclique. For short, we use $[t] = \{1,2,\ldots,t\}$ for every $t \in \mathbb{N}$. %is as the problem of finding  biclique.  

\begin{lemma}
	\label{lem:Pas}
	Assuming SSEH, for any $0<\delta < 0.1^{10}$ such that $ \delta^{-\frac{1}{2}} \in \mathbb{N}$ and $n \in \mathbb{N}$, given a bipartite graph $G = (A,B,E)$ with $|A| = n$ and $|B| = n^{t}, t = \delta^{-\frac{1}{2}}$, it is NP-hard to distinguish between the following two cases: %\footnote{For simplicity, assume w.l.o.g. that $t \in \mathbb{N}$ throughout this section.} %; this can be easily fixed by taking a suitable $\delta$ for which $t \in \mathbb{N}$.} 
\begin{itemize}
		\item 	(Completeness, "yes" case) $G$ contains $K_{\left(\frac{1}{2}-\delta\right) \cdot n,\left(1-\delta^{\frac{1}{10}}\right) \cdot n^t}$ as a subgraph. 
		
		\item (Soundness, "no" case) $G$ does not contain $K_{ \delta \cdot n,  \delta^{\frac{1}{10}} \cdot n^t}$ as a subgraph. 
		
		%Here EH(T) , {e ∈ EH | e ⊆ T} denotes the set of hyperedges inside of the set T ⊆ VH.
	\end{itemize}

\end{lemma}

We show that if we can distinguish between the two cases of Lemma~\ref{lem:Pas}, then we can also distinguish between the two cases of Lemma~\ref{lem:hardness1}. Let $0<\delta <0.1^{10}$ such that $ \delta^{-\frac{1}{2}} \in \mathbb{N}$, $n \in \mathbb{N}$,  and $G = (L,R,E)$ be a bipartite graph with $|L| = |R| = n$. %as in Lemma~\ref{lem:hardness1}. 
To reach the conditions of Lemma~\ref{lem:Pas}, we modify the graph $G$ asymmetrically, where $L$ does not change and $R$ is replaced by a set of vectors $B$, where each entry of the vectors corresponds to a vertex of $R$. The edges in the modified (bipartite) graph connect vertices in $L$ to vectors in $B$.
%edges are connected between any vertex of $L$ and any vector of $B$ with at least one entry inducing an edge in $G$.

 Specifically, let $t = \delta^{-\frac{1}{2}}$. Define $B = \{(r_1, \ldots, r_t) ~|~ \forall i \in [t]:r_i \in R\}$ as the set of vectors of length $t$, with entries that are vertices from $R$; also, let $A = L$ and define 
 \begin{equation}
	\label{eq:OR}
	\bar{E} = \{ \left(\ell, \left(r_1, \ldots, r_t\right)\right) \in A \times B~|~ \exists i \in [t] \text{ s.t. } (\ell,r_i) \in E \}. 
\end{equation}

There is an edge in $\bar{E}$ between a vector in $B$ and a vertex $\ell \in A$ if at least one of the entries in this vector is adjacent to $\ell$ in the original graph.  Finally, define the {\em reduced graph} of $G$ as $\bar{G} = (A,B,\bar{E})$. We give an example of the construction of $\bar{G}$ in Figure~\ref{fig:G}.\footnote{To simplify the example, we take $\delta>0.1$.} The proof of Lemma~\ref{lem:Pas} follows by showing that $\bar{G}$ satisfies the "yes" case of Lemma~\ref{lem:Pas} iff $G$ satisfies the "yes" case of Lemma~\ref{lem:hardness1}. %We now show the equivalence of the completeness and soundness of the two lemmas. That is, we show that if the completeness case holds for Lemma~\ref{lem:hardness1} then it also holds for this lemma, and if the  soundness case holds for Lemma~\ref{lem:hardness1} then it also holds for this lemma. Hence, distinguishing between the cases in this lemma is at least as hard as distinguishing between the cases in Lemma~\ref{lem:hardness1}. 

\begin{figure}
%	\hspace{4cm}{
\begin{tikzpicture}[scale=1.4, every node/.style={draw, circle, inner sep=1pt}]
	% first bipartite graph
	\node (l1a) at (5.5,-0.5) {$\ell_1$};
	\node (l2a) at (7,-0.5) {$\ell_2$};
	\node (r1a) at (5.5,-1.5) {$\bf \textcolor{red}{r_1}$};
	\node (r2a) at (7,-1.5) {$\bf \textcolor{blue}{r_2}$};
	\draw (l1a) -- (r1a);
	\draw (l2a) -- (r1a);
	\draw (l2a) -- (r2a);
	
	% second bipartite graph
	\node (l1b) at (5.5,2.5) {$\ell_1$};
	\node (l2b) at (7,2.5) {$\ell_2$};
	\node (r1b) at (1,1) { $\bf (\textcolor{red}{ r_1,r_1,r_1})$};
	\node (r2b) at (2.5,1) { $\bf (\textcolor{red}{ r_1,r_1},\textcolor{blue}{r_2})$};
	\node (r3b) at (4,1) { $\bf (\textcolor{red}{ r_1},\textcolor{blue}{r_2},\textcolor{red}{ r_1})$};
	\node (r4b) at (5.5,1) { $\bf (\textcolor{blue}{ r_2},\textcolor{red}{r_1,r_1})$};
	\node (r5b) at (7,1) { $\bf (\textcolor{red}{ r_1},\textcolor{blue}{r_2,r_2})$};
	\node (r6b) at (8.5,1) { $\bf (\textcolor{blue}{ r_2},\textcolor{red}{r_1},\textcolor{blue}{ r_2})$};
	\node (r7b) at (10,1) { $\bf (\textcolor{blue}{r_2,r_2},\textcolor{red}{r_1})$};
	\node (r8b) at (11.5,1) { $\bf (\textcolor{blue}{r_2,r_2,r_2})$};
	\draw (l1b) -- (r1b);
	\draw (l1b) -- (r2b);
	\draw (l1b) -- (r3b);
	\draw (l1b) -- (r4b);
	\draw (l1b) -- (r5b);
	\draw (l1b) -- (r6b);
	\draw (l1b) -- (r7b);
	\draw (l2b) -- (r8b);
	\draw (l2b) -- (r2b);
	\draw (l2b) -- (r3b);
	\draw (l2b) -- (r4b);
	\draw (l2b) -- (r5b);
	\draw (l2b) -- (r6b);
	\draw (l2b) -- (r7b);
\end{tikzpicture}
	\caption{\label{fig:G} The construction of the reduced graph $\bar{G}$ (on the top) of a graph $G$  (in the bottom)  with the vertices $L = \{\ell_1,\ell_2\}$ and $R = \{r_1,r_2\}$, with parameters $t = 3$ and $n = 2$. Observe that the biclique $K_{1,1}$ (e.g., with the vertices $\ell_1,r_1$) in $G$ induces the biclique $K_{1,7}$ in $\bar{G}$. }
\end{figure}
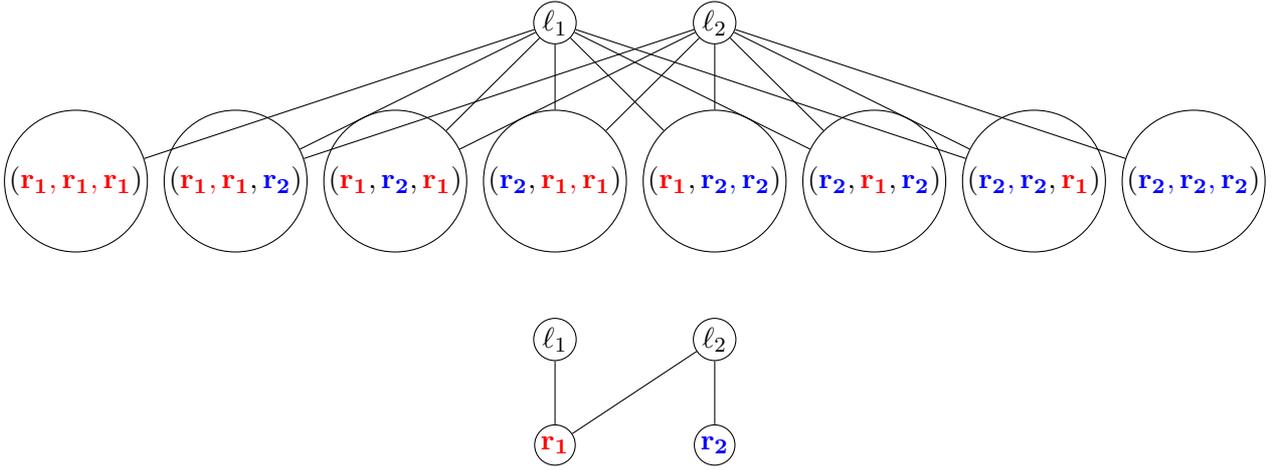

\noindent
{\bf Proof of Lemma~\ref{lem:Pas}:} 
%We show that if we can distinguish between the two cases of the lemma, then we can also distinguish between the two cases of Lemma~\ref{lem:hardness1}. Let $0<\delta <0.1$, $n \in \mathbb{N}$,  and a bipartite graph $G = (L,R,E)$ with $|L| = |R| = n$. Now, let $t = \delta^{-\frac{1}{2}}$ and define $B = \{(r_1, \ldots, r_t) ~|~ \forall i \in [t]:r_i \in R\}$.\footnote{Assume, w.l.o.g. that $t \in \mathbb{N}$.} In addition, define $A = L$. Also, let \begin{equation}
%	\label{eq:OR}
%	\bar{E} = \{ \left(\ell, \left(r_1, \ldots, r_t\right)\right) \in A \times B~|~ \exists i \in [t] \text{ s.t. } (\ell,r_i) \in E \}\}. 
%\end{equation} Finally, define the bipartite graph $\bar{G} = (A,B,\bar{E})$. 
Let $G = (L,R,E)$ be a bipartite graph with $|L| = |R| = n$ as in Lemma~\ref{lem:hardness1}, and let $\bar{G}$ be the reduced graph of $G$ as defined above for a given $0<\delta < 0.1^{10}$, where $ \delta^{-\frac{1}{2}} \in \mathbb{N}$. We now show the equivalence of the completeness ("yes" case) and soundness ("no" case) in the two lemmas. That is, we show that the "yes" case holds for $G$ according to Lemma~\ref{lem:hardness1} iff the "yes" case of Lemma~\ref{lem:Pas} holds for $\bar{G}$. Hence, distinguishing between the cases in this lemma is at least as hard as distinguishing between the cases in Lemma~\ref{lem:hardness1}.  For the first direction, we use the next auxiliary claim.

 \begin{claim}
	\label{claim:delta1}
	For all $0<\delta<0.1$ it holds that  $ 	\left(1-\left(\frac{1}{2}+\delta\right)^{\delta^{-1/2}}\right) - 	\left(1-\delta^{1/10}\right) \geq 0$.
\end{claim}
\begin{claimproof}
	Note that another way to state the inequality for $x = \delta$ is by 
	\begin{equation}
		\label{eq:in}
		\frac{\ln(x)}{10} \geq \frac{\ln (\frac{1}{2}+x)}{\sqrt{x}} .
	\end{equation} Therefore, the inequality holds if and only if the function $f(x) = 	\frac{\ln(x)}{10} - \frac{\ln (\frac{1}{2}+x)}{\sqrt{x}}$ is non-negative in the domain $(0,0.1)$. To see this, first note that $f(0.1) \geq 1.385 > 0$. Moreover, $f(x)$ is continuous in the domain $(0,0.1)$; thus, it suffices to show that the derivative of $f(x)$ satisfies $f'(x) \leq 0$ in the domain $(0,0.1)$. We note that \begin{equation}
	\label{eq:Xc}
	\begin{aligned}
		f'(x) ={} & \frac{1}{10 \cdot x}-\frac{1}{x^{1.5}+\frac{1}{2} \cdot x^{0.5}}+\frac{\ln(x+0.5)}{2\cdot x^{1.5}} \\
		\leq{} & \frac{1}{10 \cdot x}+\frac{\ln(x+0.5)}{2 \cdot x^{1.5}} \\
			\leq{} & \frac{1}{10 \cdot x}-\frac{1}{4 \cdot x^{1.5}} \\
				\leq{} & \frac{1}{10 \cdot x}-\frac{1}{10 \cdot x^{1.5}} \\
				\leq{} & 0
	\end{aligned} 
\end{equation} The second inequality holds since $\ln (0.5+x) < -0.5$ for all $x \in (0,0.1)$. The last inequality holds since $x^{1.5} \leq x$ for all $x \in (0,0.1)$. Since $f(0.1)> 0$ and by \eqref{eq:Xc} the claim follows.

%Note that the growth rate of the function $\frac{1}{\sqrt{x}}$ is significantly faster than the growth rate of the function $-\ln (x)$ for values that get closer to zero (in particular, in the domain $(0,0.1)$). In addition, it holds that $	 \frac{\ln(0.1)}{10} \approx -0.23$, and $	\frac{\ln (\frac{1}{2}+0.1)}{\sqrt{0.1}} \approx -1.615$.  Hence,  using the monotonicity of the considered functions $\ln (x), \frac{1}{\sqrt{x}}$ in $0<x<0.1$, and the known growth rates of the functions $\ln (x), \frac{1}{\sqrt{x}}$ for $0<x<0.1$, by \eqref{eq:in} the proof follows. 

% since $	 \frac{\ln(0.1)}{10} \approx -0.23$, and $	\frac{\ln (\frac{1}{2}+0.1)}{\sqrt{0.1}} \approx -1.615$. That is, for any value  
	%	Similar arguments to Claim~\ref{claim:delta2}. 
\end{claimproof}

%	\item (Completeness) 

%We now prove the "yes" case. 
We now prove the first direction, which is the "yes" case of the lemma given that $G$ satisfies the "yes" case of Lemma~\ref{lem:hardness1}. Assume that $G$ contains $K_{\left(\frac{1}{2}-\delta\right) \cdot n, \left(\frac{1}{2}-\delta\right) \cdot n}$ as a subgraph, and let $K$ be the set of vertices of such a subgraph. Let $F_A = K \cap A$ and \begin{equation}
		\label{eq:OR2}
		F_B = \{\left(r_1, \ldots, r_t\right) \in B~|~\exists i \in [t] \text{ s.t. } r_i \in K\}.
	\end{equation}  Let $a \in F_A$ and $b = (r_1, \ldots,r_t) \in F_B$; Since $K$ is a biclique, it follows by \eqref{eq:OR2} that there is $i \in [t]$ such that $(a,r_i) \in E$. Thus, by \eqref{eq:OR} it holds that $(a,b) \in \bar{E}$ and we conclude that $F_A \cup F_B$ forms a biclique in $\bar{G}$. Observe that $|F_A| \geq \left(\frac{1}{2}-\delta\right) \cdot n$, thus  
	\begin{equation}
		\label{eq:FB}
		\begin{aligned}
			|F_B| ={} & n^t-|\{(r_1, \ldots,r_t) \in B~|~\forall i \in [t]: r_i \notin K\}| \\ 
				={} & n^t-\left(n- |F_A|\right)^t  \\
			\geq{} & n^t-\left(n-\left(\frac{1}{2}-\delta\right) \cdot n\right)^t \\ 
			={} & 
			n^t-\left(\frac{1}{2}+\delta\right)^t\cdot n^t \\ ={} & 
			\left(1-\left(\frac{1}{2}+\delta\right)^{\delta^{-1/2}}\right) \cdot n^t \\ \geq{} & 
			\left(1-\delta^{0.1}\right) \cdot n^t
		\end{aligned}
	\end{equation} The first equality holds by \eqref{eq:OR2}. The first inequality holds since  $|F_A| \geq \left(\frac{1}{2}-\delta\right) \cdot n$. The last inequality follows by Claim~\ref{claim:delta1}. Hence, by \eqref{eq:FB}, it follows that $\bar{G}$ contains $K_{\left(\frac{1}{2}-\delta\right) \cdot n,\left(1-\delta^{\frac{1}{10}}\right) \cdot n^t}$ as a subgraph, which is $F_A \cup F_B$.

	%\item (Soundness) 
We now prove the second direction, which is the "no" case of the lemma given that $G$ satisfies the "no" case of Lemma~\ref{lem:hardness1}.	For this, we use the next auxiliary claim.
%	\begin{claim}
%	\label{claim:delta2}
%	For all $0<\delta<0.1$ it holds that  $ 	\delta^{0.1}-\left(1-(1-\delta)^{\delta^{-1/2}}\right) > 0$.
%\end{claim}
%\begin{claimproof}
%	Let $\delta = y^{10}$. 
%	We rewrite the inequality using the function $h(y) = \frac{\ln(1-y^{10})}{y^5} - \ln(1-y)$ and our goal becomes to show that $h(y) > 0$ for the appropriate values of $y$ in the domain $(0,1)$. %for $0<y^{10} <0.1$. 
%	Now, let \begin{equation*}
%		\label{eq:h1}
%		h_1(y) = y^6 \cdot h'(y) = y^6 \cdot \left( \frac{1}{1-y} -\frac{10 \cdot y^4}{1-y^{10}}\right)-5 \cdot \ln(1-y^{10}). 
%	\end{equation*} we now compute the derivative of $h_1$. \begin{equation*}
%	h'_1(y) = 
%\end{equation*}Then, \begin{equation}
%		\label{eq:h2}
%		\begin{aligned}
%			h'_1(y)   \cdot{} & \frac{1-y^{10}}{(1-y) \cdot y^5} = 5 \cdot y^{18}+9 \cdot y^{17}+12 \cdot y^{16}+14 \cdot y^{15}+15 \cdot y^{14}+65 \cdot y^{13}+64 \cdot y^{12}+62 \cdot y^{11} \\
%			+{} & 59 \cdot y^{10}+55 \cdot y^{9}+40 \cdot y^{8}+26 \cdot y^{7}+13 \cdot y^{6}+y^3 \cdot ( (y-5)^2+5)+21 \cdot y^{2}+13 \cdot y+6.  
%		\end{aligned}
%	\end{equation} By \eqref{eq:h2} it holds that $h'_1(\delta) >0$ for $0<\delta<0.1$, since all terms are positive in the domain. Thus, $h_1$ is an increasing function in the domain $(0,0.1)$. Moreover, $\lim_{x \rightarrow 0^+} h_1(x) = 0$. Thus, $h_1(\delta) >0$ for all $0<\delta<0.1$. This implies that $h$ is increasing, and as $\lim_{x \rightarrow 0^+} h(x) = 0$, it follows that $h(y) >0$ for all $0<\delta<0.1$ and $\delta = y^{10}$. 
%\end{claimproof}
	\begin{claim}
	\label{claim:delta2}
	For all $0<\delta<0.1^{10}$ it holds that  $ 	\delta^{0.1}-\left(1-(1-\delta)^{\delta^{-1/2}}\right) > 0$.
\end{claim}
\begin{claimproof}
	Let $\delta = y^{10}$. 
	We rewrite the inequality using the function $h(y) = \frac{\ln(1-y^{10})}{y^5} - \ln(1-y)$ and our goal becomes to show that $h(y) > 0$ for the appropriate values of $y$ in the domain $(0,0.1)$. %for $0<y^{10} <0.1$. 
	we now show that $h'(y)>0$ in the domain. \begin{equation}
		\begin{aligned}
			\label{eq:h2}
		h'(y) ={} & \frac{ \frac{-10 \cdot y^{9}}{(1-y^{10})} \cdot y^5-5 \cdot y^4 \cdot \ln (1-y^{10})}{y^{10}}+\frac{1}{1-y} \\
		\geq{} &  \frac{ \frac{-10 \cdot y^{9}}{(1-y^{10})} \cdot y^5}{y^{10}}+\frac{1}{1-y} \\
			={} &    \frac{-10}{(1-y^{10})} \cdot y^{4}+\frac{1}{1-y} \\
				\geq{} &    \frac{-10}{(1-0.1)} \cdot {0.1^{4}}+\frac{1}{1-y} \\
				>{} &  -\frac{1}{100}+\frac{1}{1-y} \\
					>{} &  -\frac{1}{100}+\frac{1}{1-0} \\
						>{} &  0 \\
		\end{aligned}
	\end{equation}
The inequalities hold since $y \in (0,0.1)$. Then, by \eqref{eq:h2} it holds that $h'(y) >0$ for $0<\delta = y^{10} <0.1^{10}$. Thus, $h(y)$ is an  increasing continuous function in the domain $(0,0.1)$. Moreover, using L'H{\^o}pital's
%L'Hôpital's 
rule
	 it holds that $\lim_{x \rightarrow 0^+} h(x) = 0$. Thus, $h(y)>0$ for all $0<\delta = y^{10}<0.1^{10}$. %This implies that $h$ is increasing, and as $\lim_{x \rightarrow 0^+} h(x) = 0$, it follows that $h(y) >0$ for all $0<\delta<0.1$ and $\delta = y^{10}$. 
\end{claimproof}

	Assume that $G$ does not contain $K_{ \delta \cdot n,  \delta \cdot n}$ as a subgraph, and let $T$ be a biclique subgraph of $\bar{G}$ such that $|T \cap A| \geq  \delta \cdot n$. To complete the proof, we now show that $|T \cap B| <  \delta^{\frac{1}{10}} \cdot n^t$, which guarantees that there is no $K_{ \delta \cdot n,\delta^{\frac{1}{10}} \cdot n^t}$ subgraph in $\bar{G}$. %Assume towards a contradiction that $|T \cap B| \geq \delta^{\frac{1}{10}} \cdot n^t$. 
	Observe that  \begin{equation}
		\label{eq:T}
		\begin{aligned}
			|T \cap B| ={} &
			|\{(r_1, \ldots,r_t) \in B~|~ \forall a \in T \cap A~  \exists i \in [t] \text{ s.t. } (a,r_i) \in E\}| \\ 
			={} &
			n^t-|\{(r_1, \ldots,r_t) \in B~|~ \exists a \in T \cap A \text{ s.t. } \forall i \in [t]: (a,r_i) \notin E \}| \\ 
			\leq{} &
			n^t-\left(n-\delta \cdot n\right)^t \\ 
			={} &
			n^t-(1-\delta)^t \cdot n^t \\ 
			={} &
			\left(1-(1-\delta)^{{\delta^{-1/2}}}\right) \cdot n^t \\ <{} &
			\delta^{0.1} \cdot n^t. 
		\end{aligned}
	\end{equation} The first equality holds by \eqref{eq:OR}. The first inequality holds since $G$ does not contain $K_{ \delta \cdot n,  \delta \cdot n}$ as a subgraph, which imply that at most $\delta \cdot n$ vertices in $R$ are connected to all vertices in $T \cap A$ in $G$; thus, there are at least $n-\delta \cdot n$ vertices $r \in R$ for which there is $a \in T \cap A = T \cap L$ such that $(a,r) \notin E$. Therefore, there are at least $(n-\delta \cdot n)^t$ vectors of length $t$ for which each entry belongs to $r \in R$ such that  $(a,r) \notin E$ for some $a \in T \cap A$. The last inequality follows by Claim~\ref{claim:delta2}. Hence, by \eqref{eq:T} it follows that $\bar{G}$ does not contain $K_{ \delta \cdot n,  \delta^{\frac{1}{10}} \cdot n^t}$ as a subgraph. Overall, we show that the "yes" case of Lemma~\ref{lem:hardness1} holds for $G$ if and only if the "yes" case of Lemma~\ref{lem:Pas} holds for $\bar{G}$. \qed
	
%\end{itemize}

%lem:hardness1

Using Lemma~\ref{lem:Pas}, we give a hardness result for U-MSB by transforming the given graph in the lemma into a U-MSB instance. Informally, given a graph $G = (A,B,E)$ as described in Lemma~\ref{lem:Pas},  we generate many {\em copies} of each vertex from $A$. Each copy of a vertex $a \in A$ is connected to all vertices $b \in B$ not connected to $a$ in $G$. The vertices of $B$ will be assigned a larger weight, making them more attractive for an MSB algorithm. %(and therefore also for MSB and MSP).	

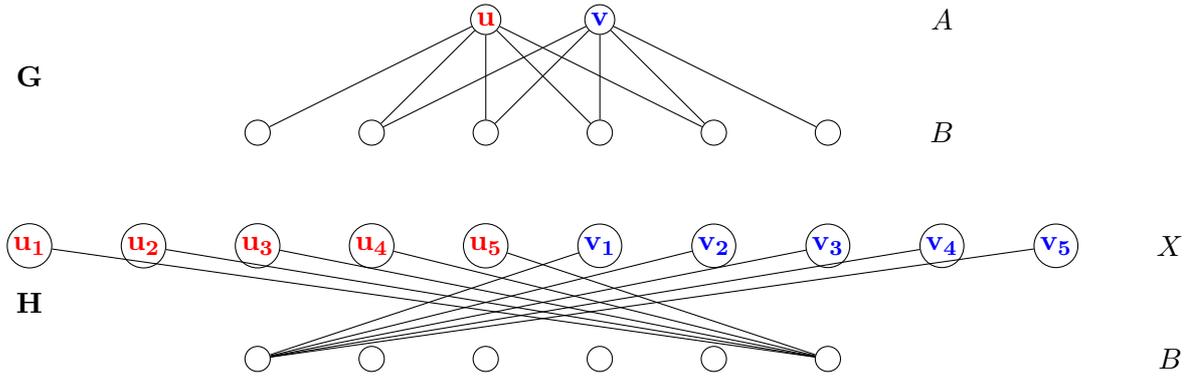
\begin{figure}
		\hspace{1cm}{
		\begin{tikzpicture}[scale=1.5, every node/.style={draw, circle, inner sep=1pt}]
			% first bipartite graph

																						    \node[draw=none] at (1,1.5) {$\bf G$};
		
  \node[draw=none] at (9,2) {$A$};
  
    \node[draw=none] at (9,1) {$B$};
			% second bipartite graph
			\node (l1b) at (5,2) {$\bf \textcolor{red}{u}$};
			\node (l2b) at (6,2) {$\bf \textcolor{blue}{v}$};
			\node (r1b) at (3,1) {$~~$};
			\node (r2b) at (4,1) {$~~$};
			\node (r3b) at (5,1) {$~~$};
			\node (r4b) at (6,1) {$~~$};
			\node (r5b) at (7,1) {$~~$};
			\node (r6b) at (8,1) {$~~$};
		%	\node (r7b) at (7,1) {\smaller \smaller \smaller $\bf (\textcolor{blue}{r_2,r_2},\textcolor{red}{r_1})$};
		%	\node (r8b) at (8,1) {\smaller \smaller \smaller $\bf (\textcolor{blue}{r_2,r_2,r_2})$};
			\draw (l1b) -- (r1b);
			\draw (l1b) -- (r2b);
			\draw (l1b) -- (r3b);
			\draw (l1b) -- (r4b);
			\draw (l1b) -- (r5b);
		%	\draw (l1b) -- (r6b);
		%	\draw (l1b) -- (r7b);
		%	\draw (l2b) -- (r8b);
			\draw (l2b) -- (r2b);
			\draw (l2b) -- (r3b);
			\draw (l2b) -- (r4b);
			\draw (l2b) -- (r5b);
			\draw (l2b) -- (r6b);
		%	\draw (l2b) -- (r7b);

		    \node[draw=none] at (11,0) {$X$};
		
			\node (u1) at (1,0) {$\bf \textcolor{red}{u_1}$};
						\node (u2) at (2,0) {$\bf \textcolor{red}{u_2}$};
									\node (u3) at (3,0) {$\bf \textcolor{red}{u_3}$};
												\node (u4) at (4,0) {$\bf \textcolor{red}{u_4}$};
															\node (u5) at (5,0) {$\bf \textcolor{red}{u_5}$};

																\node (v1) at (6,0) {$\bf \textcolor{blue}{v_1}$};
																		\node (v2) at (7,0) {$\bf \textcolor{blue}{v_2}$};
																					\node (v3) at (8,0) {$\bf \textcolor{blue}{v_3}$};
																		\node (v4) at (9,0) {$\bf \textcolor{blue}{v_4}$};
																			\node (v5) at (10,0) {$\bf \textcolor{blue}{v_5}$};

				 \node[draw=none] at (1,-0.5) {$\bf H$};
											    \node[draw=none] at (11,-1) {$B$};
																			
			\node (r1) at (3,-1) {$~~$};
		\node (r2) at (4,-1) {$~~$};
		\node (r3) at (5,-1) {$~~$};
		\node (r4) at (6,-1) {$~~$};
		\node (r5) at (7,-1) {$~~$};
		\node (r6) at (8,-1) {$~~$};
		
		\draw (r1) -- (v1);
			\draw (r1) -- (v2);
				\draw (r1) -- (v3);
					\draw (r1) -- (v4);
						\draw (r1) -- (v5);
		
			\draw (r6) -- (u1);
		\draw (r6) -- (u2);
		\draw (r6) -- (u3);
		\draw (r6) -- (u4);
		\draw (r6) -- (u5);
%		\node (l2b) at (4,2) {$\bf \textcolor{blue}{v}$};
%		\node (r1b) at (1,1) {$~~$};
%		\node (r2b) at (2,1) {$~~$};
%		\node (r3b) at (3,1) {$~~$};
%		\node (r4b) at (4,1) {$~~$};
%		\node (r5b) at (5,1) {$~~$};
%		\node (r6b) at (6,1) {$~~$};
%		%	\node (r7b) at (7,1) {\smaller \smaller \smaller $\bf (\textcolor{blue}{r_2,r_2},\textcolor{red}{r_1})$};
%		%	\node (r8b) at (8,1) {\smaller \smaller \smaller $\bf (\textcolor{blue}{r_2,r_2,r_2})$};
%		\draw (l1b) -- (r1b);
%		\draw (l1b) -- (r2b);
%		\draw (l1b) -- (r3b);
%		\draw (l1b) -- (r4b);
%		\draw (l1b) -- (r5b);
%		\draw (l1b) -- (r6b);
%		%	\draw (l1b) -- (r7b);
%		%	\draw (l2b) -- (r8b);
%		\draw (l2b) -- (r2b);
%		\draw (l2b) -- (r3b);
%		\draw (l2b) -- (r4b);
%		\draw (l2b) -- (r5b);
%		\draw (l2b) -- (r6b);
		\end{tikzpicture}
		\caption{\label{fig:H} An ilustration of the construction of the graph $H$ given the graph $G$. $X$ is constructed by creating five copies for each vertex in $A$. } %The vertices in $X$ receive a smaller weight than the vertices in $B$.}
	}
	\end{figure}

More concretely, let $\eps \in \left(0,0.1^{10}\right)$ %. We take the largest $\eps \in (0,\eps']$ 
such that $\eps^{-1} \in \mathbb{N}$ and let $n \in \mathbb{N}_{\textnormal{even}}$ (i.e., $n$ is an even integer). Define \begin{equation}
	\label{eq:val2}
	\begin{aligned}
		p =  \frac{\eps^{-3} \cdot n}{2},~~~~~~~~
		q =  \eps^{-3} \cdot n^t,~~~~~~~~
		\delta = \eps^{100}.
	\end{aligned}
\end{equation}%For some 
%Moreover, let $0<\delta < 0.1$ %whose exact value w.r.t. $\eps$ is defined later, consider 
Note that $p,q,\delta^{-\frac{1}{2}} \in \mathbb{N}$ and that $\delta \in \left(0,0.1^{10} \right)$. Let $G = (A,B,E)$ be a bipartite graph  with $|A| = n$ and $|B| = n^{t}, t = \delta^{-\frac{1}{2}}$ (as the graph described in Lemma~\ref{lem:Pas}). %For a given $\eps \in (0,0.1)$, 
We define the {\em U-MSB graph} $H = (X,B, \tilde{E})$ of $G$ by generating $q = \eps^{-3} \cdot n^t$ disjoint {\em copies} for each vertex in $A$. Formally, let $X = \{a_i~|~ i \in [q], a \in A\}$; also, let $\tilde{E} = \{(a_i,b) ~|~ i \in [q], (a,b) \notin E\}$; that is, we connect all copies of a vertex to all vertices in $B$ which are not connected by an edge to the original vertex in $A$ (see Figure~\ref{fig:H} for an illustration of the construction). We also define
for the vertex weights and costs as follows. %For some fixed $p > 1$,
 $\forall b \in B: w(b) = p$, $\forall x \in X: w(x) = 1$, and $\forall v \in X \cup B:~ c(v) = 1$. That is, the vertices of $B$ have weight $p$ that is much larger than the weight of the vertices in $X$. Also, note that the costs are uniform and we have a U-MSB instance. Finally, let the cardinality constraint (budget) be $\beta = q \cdot \frac{n}{2}+n^t$ and define the {\em reduced U-MSB instance} \begin{equation}
 	\label{eq:inst}
 	\ci = (X,B,\tilde{E},\beta,c,w).
 \end{equation} Note that $\ci$ depends on %$\delta,q,p$ and therefore on 
$\eps$, $n$, and $G$; in addition, since $n \in \mathbb{N}_{\textnormal{even}}$ it holds that $\beta \in \mathbb{N}$.  %; their definition is deferred to \eqref{eq:val2}. %of $G$ and $\eps$.  

Intuitively, let us consider the different types of solutions a U-MSB approximation algorithm $\cA$ may return for the reduced instance $\cI$. We show that the only solution that is strictly better than a $2$-approximation must take a sufficient number of vertices from each side of the graph $H$. One option for a solution is to take all vertices from $B$ (feasible under the budget constraint and is an independent set in $H$); this solution is at most a $2$-approximation for the optimum. A second option is a solution that is a subset of $X$ (without violating the budget constraint), which cannot do better than a $2$-approximation. By Lemma~\ref{lem:Pas}, we cannot expect that an algorithm for U-MSB would return a solution with many vertices from each side of the graph $H$. Hence, an approximation ratio much better than $2$
 is unlikely to exist. This gives the idea for the proof of Theorem~\ref{thm:KP}.

\noindent
{\bf Proof of Theorem~\ref{thm:KP}:}
We show the hardness result by distinguishing between the two cases in Lemma~\ref{lem:Pas}. Let $0 \leq \eps' \leq 1$, and let $$\eps = \max \{\tilde{\eps} \in \left(0,0.1^{10}\right)~|~\tilde{\eps}^{(-1)} \in \mathbb{N}, \tilde{\eps} \leq \eps'\}.$$ Observe that $\eps \in (0,0.1^{10})$, $\eps \leq \eps'$, and $\eps^{-1} \in \mathbb{N}$. Thus, %our goal is to show that there is no $2-\eps$-approximation for U-MSB. 
 if there is no $(2-\eps)$-approximation for U-MSB then there is no $(2-\eps')$-approximation for U-MSB as well. % let $0<\delta < 0.1$ %whose exact value w.r.t. $\eps$ is defined later, consider 
Let $n \in \mathbb{N}_{\textnormal{even}}$ and let $\delta = \eps^{100}$ (as in \eqref{eq:val2}). Also, let $G = (A,B,E)$ be a bipartite graph with $|A| = n$ and $|B| = n^{t}, t = \delta^{-\frac{1}{2}}$. Observe that $G$ is a graph which satisfies the conditions of Lemma~\ref{lem:Pas}. %In the proof, we use several parameters; we define their value w.r.t. $\eps$ at the end of the proof for simplicity.  

%Let $0<\delta < 0.1$, $n \in \mathbb{N}$, and a bipartite graph $G = (A,B,E)$ with $|A| = n$ and $|B| = n^{t}, t = \delta^{-\frac{1}{2}}$. We define a graph $H = (X,B, \tilde{E})$ by creating $c$ disjoint {\em copies} for each vertex in $A$. Formally, let $X = \{a_i~|~ i \in [c], a \in A\}$; also, let $\tilde{E} = \{(a_i,b) ~|~ i \in [c], (a,b) \notin E\}$; that is, we connect all copies of a vertex to all vertices in $B$ which are not connected by an edge to the original vertex in $A$.  
%
%We also define weight and cost: $\forall b \in B: w(b) = p$, $\forall x \in X: w(x) = 1$, and $\forall v \in X \cup B:~ c(v) = 1$. Finally, let the cardinality constraint (budget) be $\beta = c \cdot \frac{n}{2}+n^t$. In the above, we construct the MSB instance (with cardinality constraint rather than a more general budget) $\ci = (X,B,\tilde{E},\beta,c,w)$. 

Assume towards a contradiction that there is a $(2-\eps)$-approximation algorithm $\cA$ for U-MSB. We show below that using $\cA$ we can distinguish between the two cases of Lemma~\ref{lem:Pas} for $G$. Consider the U-MSB instance $\ci$ as defined in \eqref{eq:inst} for $\eps,n$, and $G$. Let $\OPT_y$ and $\OPT_n$ be some optimum solutions for $\ci$ in the "yes" case and in the "no" case, respectively. We show that the ratio between the optimum values of $\OPT_y$ and $\OPT_n$  (in the "yes" and "no" cases, respectively), is strictly larger than $(2-\eps)$. Hence, a $(2-\eps)$-approximation algorithm can also decide the two cases of Lemma~\ref{lem:Pas}, contradicting the complexity assumption SSEH. We consider the "no" and "yes" cases in the following auxiliary claims.

\begin{claim}
	\label{clm:yes}
	If $G$ contains $K_{\left(\frac{1}{2}-\delta\right) \cdot n,\left(1-\delta^{\frac{1}{10}}\right) \cdot n^t}$ as a subgraph (the "yes" case), then $$\OPT(\cI) = w\left(\OPT_y \right)\geq q \cdot (\frac{1}{2}-\delta) \cdot n +(1-\delta^{\frac{1}{10}}) \cdot n^t \cdot p.$$ 
\end{claim} 
\begin{claimproof}
	Assume that $G$ contains $K_{\left(\frac{1}{2}-\delta\right) \cdot n,\left(1-\delta^{\frac{1}{10}}\right) \cdot n^t}$ as a subgraph (the "yes" case) and let $K$ be such a subgraph. Let $Q = \{a_i~| i \in [q], a \in K \cap A\}$ be all copies of vertices from $K \cap A$ and let $R = Q \cup (K \cap B)$. Observe that  for all $a \in K \cap A$ and $b \in K \cap B$ it holds that $(a,b) \in E$; thus, for all $i \in [q]$ it holds that $(a_i,b) \notin \tilde{E}$, and we conclude that $R$ is an independent set in $H = (X,B,\tilde{E})$ (see Figure~\ref{fig:H}). Thus, %the optimal weight of a solution for $\ci$ is at least the weight of $R$, since 
	$R$ is a feasible solution for $\cI$, since it satisfies the budget constraint: $$ c(R) = \sum_{r \in R} c(r) = |Q|+|K \cap B| \leq q \cdot (\frac{1}{2}-\delta) \cdot n+(1-\delta^{\frac{1}{10}}) \cdot n^t \leq \frac{q \cdot n}{2}+n^t = \beta.$$  Hence, the weight of $R$ is a lower bound for the weight of the optimum of $\ci$ in the "yes" case: \begin{equation*}
		\label{eq:R}
		\begin{aligned}
			\OPT(\cI) = w\left(\OPT_y \right) \geq w(R) = w(Q)+w(K \cap B) \geq q \cdot (\frac{1}{2}-\delta) \cdot n +(1-\delta^{\frac{1}{10}}) \cdot n^t \cdot p.
		\end{aligned}
	\end{equation*}
	\end{claimproof}

\begin{claim}
	\label{clm:no}
	If $G$ does not contain $K_{ \delta \cdot n,  \delta^{\frac{1}{10}} \cdot n^t}$ as a subgraph  (the "no" case), then $$	\OPT(\cI) = w\left(\OPT_n \right) \leq \max \{q \cdot \delta \cdot n+ n^t \cdot p,  \beta+\delta^{\frac{1}{10}} \cdot n^t \cdot p\}. $$
\end{claim} 
\begin{claimproof}
	Assume that $G$ does not contain $K_{ \delta \cdot n,  \delta^{\frac{1}{10}} \cdot n^t}$ as a subgraph  (the "no" case). Recall that $\OPT_n$ is some optimal solution for $\cI$ in the "no" case. Consider the following complementary cases. \begin{enumerate}
		\item $|\OPT_n \cap B| \leq \delta^{\frac{1}{10}} \cdot n^t$. Observe that from $X$ at most $\beta = \frac{q}{2} \cdot n+n^t$ vertices can be taken without violating the cardinality constraint of $\ci$. Thus, in this case,  \begin{equation}
			\label{eq:OPT1}
			\begin{aligned}
				w(\OPT_n) = w(\OPT_n \cap X)+w(\OPT_n \cap B) \leq \beta \cdot 1+ \delta^{\frac{1}{10}} \cdot n^t \cdot p. 
			\end{aligned}
		\end{equation}

		\item  $|\OPT_n \cap B| > \delta^{\frac{1}{10}} \cdot n^t$. Since $G$ does not contain $K_{ \delta \cdot n,  \delta^{\frac{1}{10}} \cdot n^t}$ as a subgraph, for any $S \subseteq A, T \subseteq B$, $|T| > \delta^{\frac{1}{10}} \cdot n^t$ such that $T \cup S$ is a biclique, it holds that $|S| \leq \delta \cdot n$. Hence, by the definition of $\tilde{E}$, there can be copies of at most $\delta \cdot n$ distinct vertices from $A$ in $\OPT_n$ in this case. Thus, since each vertex has $q$ copies:%, and we get:  
		
		\begin{equation}
			\label{eq:OPT2}
			\begin{aligned}
				w(\OPT_n) = w(\OPT_n \cap X)+w(\OPT_n \cap B) \leq q \cdot \delta \cdot n+ n^t \cdot p. 
			\end{aligned}
		\end{equation} 
	\end{enumerate} By \eqref{eq:OPT1} and \eqref{eq:OPT2} we get: 	\begin{equation*}
		\label{eq:OPT3}
		\begin{aligned}
		\OPT(\cI) = w(\OPT_n) \leq \max \{q \cdot \delta \cdot n+ n^t \cdot p,  \beta+\delta^{\frac{1}{10}} \cdot n^t \cdot p\}. 
		\end{aligned}
	\end{equation*} 
%\end{enumerate}
\end{claimproof}
%\begin{itemize}
	%\item 	(Completeness) 
	
%	\begin{itemize}
%		\item "no" case. 	
%		
%		\item "yes" case. Now, 

	%\item (Soundness) 
	
%\end{itemize}

%We give the following values for our parameters to achieve the desired approximation gap between the two cases: \begin{equation}
%	\label{eq:val2}
%	\begin{aligned}
%		p =  \frac{\eps^{-3} \cdot n}{2},~~~~~~~~
%		c =  \eps^{-3} \cdot n^t,~~~~~~~~
%		\delta = \eps^{100}.
%	\end{aligned}
%\end{equation} 
%Now, we can get the following inequalities, where we use 
To complete the proof of the theorem, we first simplify the notation. Let $\sigma =  \frac{\eps^{-3} \cdot n^{\eps^{-50}+1}}{2}$. Using \eqref{eq:val2}, we can bound the values of the expressions above w.r.t. $\sigma$: 
\begin{equation}
	\label{eq:val0}
	\begin{aligned}
		\beta ={} &  q \cdot n \cdot \frac{1}{2}+n^t = \frac{\eps^{-3} \cdot n^{\eps^{-50}+1}}{2} +n^t = \sigma+n^{\eps^{-50}} \leq \sigma+2\eps^3 \cdot \sigma \\
		\\
		n^t \cdot p ={} &   n^{\eps^{-50}}  \cdot \frac{\eps^{-3} \cdot n}{2} = \sigma\\
		\\
	q \cdot \delta \cdot n={} &  \eps^{-3} \cdot n^{\eps^{-50}+1} \cdot \eps^{100} = 2 \eps^{97} \cdot \sigma \leq \eps^{3} \cdot \sigma.
	%\\
%		\\
%		c \cdot \delta \cdot n+ n^t \cdot p \leq {} & \sigma + \eps^{3} \cdot \sigma \\
%		\\
%		\beta+\delta^{\frac{1}{10}} \cdot n^t \cdot p \leq{} & \sigma+2\eps^3 \cdot \sigma+\eps^{10} \cdot \sigma < (1+\eps^2) \cdot \sigma \\
%		\\
%		c \cdot (\frac{1}{2}-\delta){} & \cdot n \cdot 1 +(1-\delta^{\frac{1}{10}}) \cdot n^t \cdot p \geq \sigma -\eps^3 \cdot \sigma +\sigma -\eps^{10} \cdot \sigma \geq (2-\eps^2) \cdot \sigma
	\end{aligned}
\end{equation} 

Therefore, by \eqref{eq:val0} we have
\begin{equation}
	\label{eq:val}
	\begin{aligned}
		q \cdot \delta \cdot n+ n^t \cdot p \leq {} & \sigma + \eps^{3} \cdot \sigma \\
		\\
		\beta+\delta^{\frac{1}{10}} \cdot n^t \cdot p \leq{} & \sigma+2\eps^3 \cdot \sigma+\eps^{10} \cdot \sigma < (1+\eps^2) \cdot \sigma \\
		\\
		q \cdot (\frac{1}{2}-\delta) \cdot n +(1-\delta^{\frac{1}{10}}) \cdot n^t \cdot p \geq{} & \sigma -\eps^3 \cdot \sigma +\sigma -\eps^{10} \cdot \sigma \geq (2-\eps^2) \cdot \sigma
	\end{aligned}
\end{equation}

Finally, by Claim~\ref{clm:yes}, Claim~\ref{clm:no}, and \eqref{eq:val} we have: \begin{equation}
	\label{eq:final}
	\begin{aligned}
		\frac{w\left(\OPT_y\right)}{w\left(\OPT_n\right)} \geq \frac{(2-\eps^2) \cdot \sigma}{(1+\eps^2) \cdot \sigma} > 2-\eps. 
	\end{aligned}
\end{equation} The second inequality holds since $0<\eps<0.1$. Hence, by \eqref{eq:final} the $(2-\eps)$-approximation $\cA$ for U-MSB can always distinguish between the "yes" and "no" cases for $G$, contradicting Lemma~\ref{lem:Pas}.  \qed

\section{A Matching Upper Bound for MSP}
\label{sec:3}

%We now 
In this section, we give a $2$-approximation algorithm for MSP, matching the $(2-\eps)$ lower bound shown in Theorem~\ref{thm:KP}. Let $\ci = (V,E,w,c,B)$ be an MSP instance and  let $W(\ci) = \max_{v \in I} w(v)$ be the maximum weight of a vertex of $V$. 
We use the results of \cite{kulik2021lagrangian} designed for a general family of subset selection problems under a budget constraint: 
the main approach of~\cite{kulik2021lagrangian} is to
%their approach 
exploit the {\em Lagrangian relaxation} of the underlying problem (i.e., relax the budget constraint) and to use an algorithm for the unbudgeted version of the problem. The following is a compact representation of the relevant theorem in \cite{kulik2021lagrangian}. 
\begin{lemma}
	\label{thm:ariel} \cite{kulik2021lagrangian}
	Given an \textnormal{MSP} instance $\ci = (V,E,w,c,B)$ and $\eps>0$, such that the maximum weight independent set problem admits a polynomial-time $\rho$-approximation on $G=(V,E)$ for some $\rho \in (0,1]$,
	there is an algorithm \textnormal{\textsf{Lagrangian}} that
	 returns in time $O(\log(W(\ci))+\log(B)+\log(\eps^{-1}) \cdot \textnormal{poly}(|\ci|))$ a solution for $\ci$ of weight at least $\left(  \frac{\rho}{1+\rho}-\eps\right) \cdot \OPT(\ci)$.
\end{lemma}
We use Lemma~\ref{lem:grot} and apply the results of \cite{kulik2021lagrangian} with a suitable choice for $\eps$, which gives a good trade-off between running time and approximation guarantee. The pseudocode is given in Algorithm~\ref{alg:lag}. %We give the proof of Theorem~\ref{lem:MSP} in Appendix~\ref{sec:proofsKP}.
 \begin{algorithm}[h]
 	\label{alg:1}
	\caption{$\textsf{MSP-Solve}(\ci = (V,E,w,c,B))$}
	\label{alg:lag}
	
	Let $\eps = \frac{1}{8 \cdot |V| \cdot W(\ci)}$.
	
	Return $\textsf{Lagrangian}(\ci,\eps)$. 
\end{algorithm}

 \begin{lemma}
	\label{lem:MSP}
	Algorithm~\ref{alg:lag} is a $2$-approximation for \textnormal{MSP}.
\end{lemma}

\begin{proof}
%	\noindent
%	{\bf Proof of Theorem~\ref{lem:MSP}:}
	Let  $\ci = (I,E,w,c,B)$ be an MSP instance. By scaling, we may assume that the input parameters are all integers and therefore $\OPT(\ci) \in \mathbb{N}$. By Lemma~\ref{thm:ariel}, it holds that $S = \textsf{MSP-Solve}(\ci)$ is a solution of $\ci$ of weight at least \begin{equation}
		\label{eq:OPtimum}
		\begin{aligned}
			w(S) \geq{} & \left( \frac{\rho}{1+\rho}-\eps\right) \cdot \OPT(\ci) \\
			\geq{} & \left( \frac{1}{1+1}-\eps\right) \cdot \OPT(\ci) \\
			={} &  \frac{\OPT(\ci)}{2} - \frac{\OPT(\ci)}{8 \cdot |V| \cdot W(\ci)} \\
			\geq{} &  \frac{\OPT(\ci)}{2} - \frac{1}{8} \\
		\end{aligned}
	\end{equation} The second inequality holds because $\rho = 1$ by Lemma~\ref{lem:grot} (i.e., MWIS can be solved in polynomial time on perfect graphs).
	%weighted independent sets in perfect graphs can be found in polynomial time). 
	The equality follows by the selection of $\eps$. The last inequality holds since $|V|\cdot W(\ci)$ is a trivial upper bound on $\OPT(\ci)$. Recall that $\OPT(\ci) \in \mathbb{N}$. Therefore, if $\OPT(\ci)$ is even, it holds that $ \frac{\OPT(\ci)}{2}-\frac{1}{8}> \ceil{\frac{\OPT(\ci)}{2}}-1$; otherwise, $\OPT(\ci)$ is odd and it holds that $ \frac{\OPT(\ci)}{2}-\frac{1}{8} =  \frac{\OPT(\ci)-1}{2}+\frac{1}{2} - \frac{1}{8} > \ceil{\frac{\OPT(\ci)}{2}}-1$. Hence, by \eqref{eq:OPtimum} it holds that $w(S) \geq \frac{\OPT(\ci)}{2}$. Finally, by the selection of $\eps$ it holds that the running time of the algorithm is polynomial in $|\ci|$, the encoding size of the instance.%\qed
\end{proof}

\noindent {\bf Proof of Theorem~\ref{thm:MSP}:} The proof follows immediately from Lemma~\ref{lem:MSP}. \qed

\section{Capacitated MWIS in Bipartite Graphs}
\label{sec:cap}

%	\begin{lemma}
%	\label{lem:EPTAS}
%	Unless \textnormal{FPT=W[1]}, there is no \textnormal{EPTAS} for \textnormal{CMWIS} in bipartite graphs.\footnote{For more details on the relevant parametrized complexity assumptions, see, e.g., \cite{cygan2015parameterized}.}
%\end{lemma}

In this section we give the proof of Theorem~\ref{lem:EPTAS} using a reduction from the balanced $k$-biclique problem parametrized by $k$ - the size of the maximal biclique in each side. The basic idea of the proof is as follows. Given a bipartite graph $G = (L,R,E)$ and the parameter $k$ we define a CMWIS instance on a bipartite graph $H$ having the complement set of edges to $E$ between $L$ and $R$.  We define the weight (or, the cost) and the budget such that there can be at most $2k$ vertices in a solution to the CMWIS instance; we also define the the weights asymmetrically between the two sides of $H$, such that an optimal solution to the CMWIS instance has exactly $k$ vertices from each side of $H$ if and only if there is a biclique in the original graph $G$.

\comment{	\begin{figure}
		\centering
		\vspace{-4cm} 
		\begin{tikzpicture}[scale=1.4, every node/.style={draw, circle, inner sep=1pt}]
			\node (11) at (5.5,4) {$\bf \textcolor{teal}{1}$};
			\node (22) at (7,4) {$\bf \textcolor{black}{2}$};
			\node (33) at (5.5,3) {$\bf \textcolor{red}{3}$};
			\node (44) at (7,3) {$\bf \textcolor{blue}{4}$};
			\draw (11) -- (22);
			\draw (11) -- (33);
			\draw (11) -- (44);
			\draw (22) -- (44);
			\draw (33) -- (44);
			
			\node[draw=none] at (9, 3.5) {$~\Pi = \{\textnormal{independent sets of } G\}$};
			
			\node[draw=none] at (4.5, 3.5) {A graph $G$};
			% Formula
			\node[draw=none] at (6.75, 2.5) {$\mathcal{J}_{4,2} = \{\emptyset, \{$\bf \textcolor{teal}{1}$\},\{$\bf \textcolor{black}{2}$\},\{$\bf \textcolor{red}{3}$\},\{$\bf \textcolor{blue}{4}$\}\}, ~~\mathcal{K}_{4,2,5} = \{\{$\bf \textcolor{black}{2}$,$\bf \textcolor{blue}{4}$\},\{$\bf \textcolor{black}{2}$,$\bf \textcolor{teal}{1}$\},\{$\bf \textcolor{blue}{4}$,$\bf \textcolor{red}{3}$\},\{$\bf \textcolor{red}{3}$,$\bf \textcolor{teal}{1}$\}\}, ~~\mathcal{L}_{4,2,5}(\Pi) = \{\{$\bf \textcolor{black}{2}$,$\bf \textcolor{red}{3}$\}\}$};
		\end{tikzpicture}
		\vspace{-7.5cm} 
		\caption{\label{fig:Pi} The independent sets of the $\Pi$-matroid $M_{n,k,\alpha}(\Pi)$, with parameters $n = 4$, $k = 2$, $\alpha = 5$. The secret set $\Pi$ contains all independent sets in the graph $G$, where $\{2,3\}$ is the only independent set in $G$ with $k$ elements.}
		
\end{figure}}

{\bf Proof of Theorem~\ref{lem:EPTAS}:}  
	We give a reduction from the $k$-biclique problem in bipartite graphs. In this problem, we are given a bipartite graph $G = (L,R,E)$ and a number $k \in \mathbb{N}_{>0}$; the goal is to decide if there is a balanced biclique $K_{k,k}$ in $G$. Let $G = (L,R,E)$  and $k$ be instance $U$ of the $k$-biclique problem in bipartite graphs.  We define the following CMWIS instance $I = (L,R,\bar{E},w,B)$ such that the following holds. \begin{enumerate}
	%\item There are $n$ elements $E = \{e_1, \ldots, e_n\}$.
	\item The edges are the bipartite-complement of $E$, that is $\bar{E} =  (L \times R) \setminus E$. 
	\item For all $\ell \in L$ define the weight as  $w(\ell) = 4k^2+2 \cdot k$.
		\item For all $r \in R$ define the weight as $w(r) = 4k^2+1$.
		\item Define the budget as $B = 8k^3 +2 \cdot k^2+k$. 
\end{enumerate} Let $H = (L,R,\bar{E})$ be the induced graph. We give an example of the construction in Figure~\ref{fig:X}. We use the following auxiliary claims. 

\begin{figure}
	%	\hspace{4cm}{
		\centering
		\begin{tikzpicture}[scale=1.4, every node/.style={draw, circle, inner sep=1pt}]
			% first bipartite graph
			\node (l1a) at (5.5,-0.5) {$\bf \textcolor{red}{\ell_2}$};
			\node (1) at (4,-0.5) {$\ell_1$};
			\node (l2a) at (7,-0.5) {$\bf \textcolor{red}{\ell_3}$};
			\node (2) at (4,-1.5) {$r_1$};
			
			\node (r1a) at (5.5,-1.5) {$\bf \textcolor{blue}{r_2}$};
			\node (r2a) at (7,-1.5) {$\bf \textcolor{blue}{r_3}$};
			\draw (l1a) -- (r1a);
			\draw (l2a) -- (r1a);
			\draw (l2a) -- (r2a);

			\draw (r2a) -- (l1a);
			\draw (2) -- (l1a);
			\draw (1) -- (r1a);

			\node (L2) at (10.5,-0.5) {$\bf \textcolor{red}{\ell_2}$};
			\node (L1) at (9,-0.5) {$\ell_1$};
			\node (L3) at (12,-0.5) {$\bf \textcolor{red}{\ell_3}$};
			\node (R1) at (9,-1.5) {$r_1$};
			
			\node (R2) at (10.5,-1.5) {$\bf \textcolor{blue}{r_2}$};
			\node (R3) at (12,-1.5) {$\bf \textcolor{blue}{r_3}$};
			
			\draw (L1) -- (R1);
			\draw (L3) -- (R1);
			\draw (L1) -- (R3);
			\node[draw=none] at (5.5, -2) {$G = (L,R,E), k = 2$};
			
			\node[draw=none] at (10.5, -2) {$H = (L,R,\bar{E})$};
			
			\node[draw=none] at (13, -1.5) {$w(r_i) = 17$};
			
			\node[draw=none] at (13, -0.5) {$w(\ell_i) = 20$};
			
			\node[draw=none] at (13, -1) {$B = 74$};
		\end{tikzpicture}
		\vspace{-1.5cm} 
		\caption{\label{fig:X} The construction of the reduced CMWIS instance (on the left) given the balanced biclique instance (on the right). The maximum balanced biclique in $G$ is $\{\ell_2,\ell_3,r_2,r_3\}$ (highlighted in red and blue); this is the only solution for the reduced CMWIS instance with weight exactly $B = 74$.}
	\end{figure}
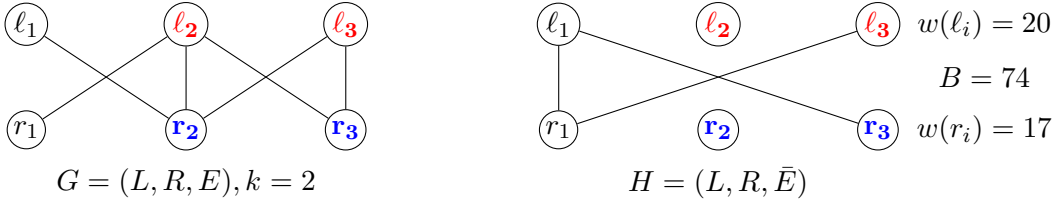

\begin{claim}
	\label{claim:1}
If there is a balanced biclique $K_{k,k}$ in $G$ then there is a solution for $I$ of weight $B$. 
\end{claim}
\begin{claimproof}
	Let $(L',R') \subseteq L \times R$ be a balanced biclique of size $k$ from each side (i.e., $K_{k,k}$) in $G$. Define $S = L' \cup R'$. Observe that $S$ is an independent set in $H$. In addition, %In addition, let $D = \{x_i~|~i \in S\}$. 
	$$w(S) = w(L')+w(R') = 8k^3+k \cdot (2k)+k \cdot  1 = 8k^3+2 \cdot k^2+k = B.$$ By the above, $S$ is a solution for $I$ of weight exactly $B$.

	%	Then, $$c(D) = p(D) = \sum_{i \in S} \left( x_i+2 \cdot \sum_{i \in [n]} x_i  \right) = T+|S| \cdot 2 \cdot \sum_{i \in [n]} x_i = T+2 k\cdot \sum_{i \in [n]} x_i = B.$$ By the above and since $D \in \cm_k$, $D$ is a solution for $I$ of profit $B$. 
\end{claimproof}
\begin{claim}
	\label{claim:2}
	If there is a solution for $I$ of weight at least $B$ then there is a balanced biclique $K_{k,k}$ in $G$. 
\end{claim}
\begin{claimproof}
	
 For any $S \subseteq L \cup R$ such that $|S|> 2k$ it holds that $w(S) > B$:
 $$w(S) \geq (2k+1) \cdot (4k^2+1) = 8k^3+4k^2+2k+1> 8k^3 +2 \cdot k^2+k = B.$$
  Moreover,  for any $S \subseteq L \cup R$ such that $|S|< 2k$ it holds that $w(S)< B$:
  $$w(S) \leq (2k-1) \cdot (4k^2+2k) = 8k^3-4k^2+4k^2-2k = 8k^3 -2k < 8k^3 +2 \cdot k^2+k = B.$$
   Note that there cannot be a solution for $I$ of weight strictly larger than $B$ by the capacity constraint. Thus, we conclude that a solution $S$ for $I$ of weight at least $B$ satisfies $|S| = 2k$. Let $S$ be a solution for $I$ of weight $B$. %Since $S$ is a solution for $I$ it holds that $S$ is independent set in $H$.  
    We now show that $S$ must be a balanced biclique. If $|S \cap L| < k$, then $w(S)< B$:
    \begin{equation}
    	\label{eq:s1}
    	\begin{aligned}
    	w(S) ={} & w(S \cap L)+w(S \cap R) \\
    	\leq{} & (k-1) \cdot (4k^2+2k)+(k+1) \cdot (4k^2+1)\\
    	={} & 4k^3-4k^2+2k^2-2k+4k^3+4k^2+k+1\\
    	={} &8k^3+2k^2-k+1\\
    		<{} &8k^3+2k^2+k\\
    		={} & B.
    	\end{aligned}
    \end{equation} The first inequality holds since $|S| = 2k$ because the weight of $S$ is $w(S) = B$; moreover, the expression $w(S)$ is maximized if there is a maximum number of vertices in $S$ from $L$. Alternatively, if $|S \cap L|>k$ then $w(S)>B$:
 \begin{equation}
	\label{eq:s2}
	\begin{aligned}
		w(S) ={} & w(S \cap L)+w(S \cap R) \\
		\geq{} & (k+1) \cdot (4k^2+2k)+(k-1) \cdot (4k^2+1)\\
		={} & 4k^3+4k^2+2k^2+2k+4k^3-4k^2+k-1\\
		={} &8k^3+2k^2+3k-1\\
		\geq{} &8k^3+2k^2+2k\\
				>{} &8k^3+2k^2+k\\
		={} & B.
	\end{aligned}
\end{equation} By \eqref{eq:s1} and \eqref{eq:s2} we conclude that $|S \cap L| = |S \cap R| = k >0$. Since $S$ is a solution for $I$ it holds that $S$ is independent set in $H$. By the definition of $H$, it implies that $S$ is a balanced biclique in $G$ of size $k$ from each side of the graph (i.e., $K_{k,k}$).

%     and if $|S \cap L| > k$ then $w(S) > B$. Thus, a solution for $I$ of weight $B$ implies a solution for $U$. 
	
%	Let $S$ be a solution for $I$ of weight at least $B$. Since $S$ is a solution of $I$ it holds that $w(S) \leq B$; thus, since $w(S) \geq B$ we conclude that $w(S) = B$. Let $L' = S \cap L$ and $R' = S \cap R$. For the following, we show that $S = L' \cup R'$ is a biclique $K_{k,k}$ in $G$. First, Assume towards a contradiction that $|L'| \neq k$. If $|S|< k$, then $$p(F) = \sum_{i \in F} i+|F| \cdot 2 \cdot \sum_{i \in [n]} i \leq \sum_{i \in F} i+(k-1) \cdot 2 \cdot \sum_{i \in [n]} i \leq 2 k\cdot \sum_{i \in [n]} i < B.$$
%	This is a contradiction that $p(F) = B$. Alternatively, if $|F|>k$: 
%	% \ariel{I changed the expression in blue - Ilan- please verify}
%	$$c(F) = \sum_{i \in F} i+|F| \cdot 2 \cdot \sum_{i \in [n]} i \geq \sum_{i \in \textcolor{blue}{F}} i+(k+1) \cdot 2 \cdot \sum_{i \in [n]} i \geq 2 \sum_{i \in [n]} i +2 k\cdot \sum_{i \in [n]} i > B.$$ This is also a contradiction since $F$ is a solution for $I$. Thus, $|F| = k$.  Therefore,
%	$$\sum_{i \in F} i = c(F) - |F| \cdot 2 \cdot \sum_{i \in [n]} i = c(F) -  2 k \cdot \sum_{i \in [n]} i = B - 2 k \cdot \sum_{i \in [n]} i = T.$$
	%
\end{claimproof}
By Claim~\ref{claim:1} and Claim~\ref{claim:2}, there is a balanced biclique $K_{k,k}$ in $G$ if and only if there is a solution for $I$ of weight at least $B$. In addition, note that the construction of $I$ given $U$ can be computed in polynomial time in the encoding size of $U$. Therefore, if there is an EPTAS $\cA$ for CMWIS on bipartite graphs it can be used to decide $U$ in FPT time as explained below. Let $\eps = \frac{1}{12 \cdot k^3}$ be an error parameter. We compute $\cA$ on $I$ and $\eps$; let $S$ be the returned solution. Since $\cA$ is an EPTAS, it returns a $(1+\eps)$-approximation for $I$: \begin{equation}
	\label{eq:WS}
	w(S) \geq (1+\eps) \cdot \OPT(I) = \OPT(I)-\frac{\OPT}{12 \cdot k^3}> \OPT(I)-\frac{\OPT(I)}{B} \geq \OPT(I)-1.
\end{equation} Since $w(S) \in \mathbb{N}$, by \eqref{eq:WS} it holds that $w(S) = \OPT(I)$. Thus, we can decide $U$ by returning that there is a balanced biclique $K_{k,k}$ in $G$ if and only if $w(S) = B$. Note that the running time of computing $S$ is $f(\frac{1}{\eps}) \cdot |I|^{O(1)}$, where $|I|$ is the encoding size of $I$ and $f$ is some computable function. Since the construction of $I$ is polynomial in $|U|$ (the encoding size of $U$) and $\eps = \frac{1}{12 \cdot k^3}$, the running time of deciding $U$ is bounded by $f(12 \cdot k^3) \cdot |U|^{O(1)}$. Since the $k$-biclique problem in bipartite graphs is known to be W[1]-Hard \cite{lin2014parameterized}, we conclude that CMWIS in bipartite graphs is also W[1]-Hard.  \qed

%Assume that FPT$\neq$W[1] and let $\cA$ be an EPTAS for CMWIS in bipartite graphs.  

\section{Discussion}
\label{sec:discussion}
In this paper we showed that the budgeted maximum weight independent set (BMWIS) problem admits a tight $2$ approximation. Our main result is a lower bound of $2-\eps$ already for the special case of a bipartite graph with uniform costs, where the previous lower bound was strong NP-hardness \cite{pferschy2009knapsack}. We also showed that there is a tight $2$-approximation for the BMWIS problem on perfect graphs, using a techniques of \cite{kulik2021lagrangian}. This resolves the complexity status of the considered problems. In fact, to the best of our knowledge, our results give the first example for which the technique of  \cite{kulik2021lagrangian} yields a tight approximation. Our paper also shows that the capacitated maximum weight independent set (CMWIS) in bipartite graphs is unlikely to admit an EPTAS; this gives a tight lower bound as CMWIS on perfect graphs admits a PTAS. 

Our hardness result may have implications for other problems. For example, in many packing problems approximation algorithms often rely on linear programs with an exponential number of variables, called {\em configuration LPs} (see, e.g., \cite{bansal2014improved,fleischer2011tight}). The bottleneck of solving the standard configuration LP for {\em bin packing with a conflict graph} 
is a BMWIS problem with the same conflict graph. Thus, our results suggest that the standard configuration LP is unlikely to yield better approximations for bin packing with bipartite \cite{huang2023improved} and perfect conflict graphs \cite{doron2023approximating}.  

For other graph families, the complexity status of BMWIS remains open. In particular, similar to Algorithm~\ref{alg:1}, we can use the techniques of \cite{kulik2021lagrangian} for other graph classes such as $d$-claw free graphs and graphs of bounded degree $\Delta$, applying as a black box the MWIS algorithms for these graph classes \cite{neuwohner2021improved,halldorsson1998approximations}. 
As $d$ or $\Delta$ grow large, the approximation guarantees for both graph classes approach the known guarantees for MWIS on these graphs~\cite{neuwohner2021improved,halldorsson1998approximations}. 
However, for small values of $d$ and $\Delta$ (e.g., $d = 4$ and $\Delta = 3$), there is a significant gap between the approximation guarantee of the MWIS algorithms and the above BMWIS algorithms. %$\frac{8}{3}$-approximation and a $\frac{5}{2}$-approximation for $d = 4$-claw free graphs and graphs of bounded degree $\Delta = 3$, respectively. 
%However, for the corresponding non-budgeted MWIS problems, there is a $\frac{5}{3}$-approximation and a $\frac{3}{2}$-approximation, respectively. 
It would be interesting to bridge these gaps, either by using stronger approximation algorithms or by providing better lower bounds. %than the bounds for the MWIS variants. 
%There are graph families, it is still an open question whther the BMWIS 

\noindent {\bf Acknowledgments:} We thank Magnus Halld{\'o}rsson for helpful comments and suggestions. We also thank Pasin Manurangsi for a personal communication which led to the statement of %result in 
%inpiring the result in
%that inspired 
Lemma~\ref{lem:Pas}.  
\bibliographystyle{splncs04}
%\newpage
\bibliography{bibfile}

%\appendix

%\input{FFD}
%\input{proofsPrel}
%\input{ProofsBPC}
%\input{Proofs53}
%\input{proofsBPHardness}
%\input{proofsMSP}
%\input{ProofsSanta}

\end{document}